\documentclass[11pt]{amsart}
\usepackage{graphicx,amssymb,amsmath,amsthm}
\usepackage{enumerate}
\usepackage{dsfont}
 \numberwithin{equation}{section}
\usepackage{mathrsfs}
\usepackage{epsfig}
\usepackage{lscape}
\usepackage{subfigure}
\usepackage{epstopdf}
\usepackage{color}
\usepackage{caption}
\usepackage{algorithm}
\usepackage{algpseudocode}

\newcommand{\xkh}[1]{\left(#1\right)}
\newcommand{\dkh}[1]{\left\{#1\right\}}

\newcommand{\jkh}[1]{\left\langle#1\right\rangle}
\newcommand{\nj}[1]{\langle {#1} \rangle}

\newcommand{\norm}[1]{\|{#1}\|_2}

\newcommand{\normone}[1]{\|{#1}\|_1}
\newcommand{\norms}[1]{\|{#1}\|}
\newcommand{\abs}[1]{\lvert#1\rvert}

\newcommand{\argmin}[1]{\mathop{\rm argmin}\limits_{#1}}

\newcommand{\s}{{\mathrm s}}

\newcommand{\E}{{\mathbb E}}

\newcommand{\PP}{{\mathbb P}}

\newcommand{\R}{{\mathbb R}}
\newcommand{\Rd}{{\mathbb R}^d}

\newcommand{\T}{\top}
\newcommand{\C}{{\mathbb C}}

\newcommand{\x}{{\widehat{\mathbf{x}}}}
\newcommand{\xx}{{\widehat{\mathbf{x}'}}}
\newcommand{\hy}{{\widehat{\mathbf{y}}}}

\newcommand{\vx}{{\mathbf x}}
\newcommand{\vy}{{\mathbf y}}
\newcommand{\vu}{{\mathbf u}}

\newcommand{\vz}{{\mathbf z}}

\newcommand{\ve}{{\mathbf \eta}}
\newcommand{\vve}{{\mathbf e}}
\newcommand{\ta}{{\tilde{\mathbf{a}}}}
\newcommand{\va}{{\mathbf a}}
\newcommand{\vh}{{\mathbf h}}

\newcommand{\F}{{\mathbb F}}

\newcommand{\supp}{{\rm supp}}

\renewcommand{\omega}{\eta}

\newcommand{\RNum}[1]{\uppercase\expandafter{\romannumeral #1\relax}}

\newtheorem{definition}{Definition}[section]

\newtheorem{theorem}[definition]{Theorem}
\newtheorem{lemma}[definition]{Lemma}

\newtheorem{remark}[definition]{Remark}

\date{}

\begin{document}
\baselineskip 18pt
\bibliographystyle{plain}
\title{The estimation performance of nonlinear least squares for phase retrieval}
\author{Meng Huang}
\address{LSEC, Inst.~Comp.~Math., Academy of
Mathematics and System Science,  Chinese Academy of Sciences, Beijing, 100091, China
\newline
School of Mathematical Sciences, University of Chinese Academy of Sciences, Beijing 100049, China
}
\email{hm@lsec.cc.ac.cn}

\author{Zhiqiang Xu}
\thanks{Zhiqiang Xu was supported  by NSFC grant (91630203, 11688101),
Beijing Natural Science Foundation (Z180002).}
\address{LSEC, Inst.~Comp.~Math., Academy of
Mathematics and System Science,  Chinese Academy of Sciences, Beijing, 100091, China
\newline
School of Mathematical Sciences, University of Chinese Academy of Sciences, Beijing 100049, China}
\email{xuzq@lsec.cc.ac.cn}
\begin{abstract}
Suppose that $\vy=\abs{A\vx_0}+\eta$ where $\vx_0\in \R^d$ is the target signal and $\eta\in \R^m$ is a noise vector.
The aim of phase retrieval is to estimate $\vx_0$ from $\vy$. A popular model for estimating $\vx_0$
is the nonlinear least square  $\x:={\rm argmin}_{\vx} \|\abs{A\vx}-\vy\|_2$.
One already develops many efficient algorithms for solving the model, such as the
seminal error reduction algorithm.
In this paper, we present the estimation performance of the  model with  proving  that $\|\x-\vx_0\|\lesssim {\|\eta\|_2}/{\sqrt{m}}$ under the assumption of $A$ being a Gaussian random matrix.
We also prove the reconstruction error  ${\|\eta\|_2}/{\sqrt{m}}$ is sharp.
  For the case where $\vx_0$ is sparse, we
study the estimation performance of both the nonlinear Lasso of phase retrieval and its unconstrained version.
Our results are non-asymptotic, and we do not assume any distribution
on the noise $\eta$.
To the best of our knowledge, our results represent the first theoretical guarantee for
the nonlinear least square and for the nonlinear Lasso of phase retrieval.

\end{abstract}
\maketitle
\section{Introduction}
\subsection{Phase retrieval}
Suppose that $\vx_0\in \F^d$ with $\F\in \{\R,\C\}$ is the target signal. The information that we gather about $\vx_0$ is
\[
 \vy=\abs{A\vx_0}+\eta,
\]
where $A=(\va_1,\ldots,\va_m)^T\in \F^{m\times d}$ is the known measurement matrix and $\eta\in \R^m$ is a noise vector. Throughout  this paper, we often assume  that $A\in \R^{m\times d}$ is a Gaussian random matrix with entries $a_{jk}\sim N(0,1)$ with $m\gtrsim d$ and we also assume that
$\eta$ is either fixed or random and independent of $A$.

The aim of phase retrieval is to estimate $\vx_0$ from $\vy$.
Phase retrieval is raised in numerous applications such as X-ray crystallography \cite{harrison1993phase,millane1990phase}, microscopy \cite{miao2008extending}, astronomy \cite{fienup1987phase}, coherent diffractive imaging \cite{shechtman2015phase,gerchberg1972practical} and optics \cite{walther1963question} etc.
A popular model for recovering $\vx_0$ is
\begin{equation}\label{eq:mod1}
\argmin{\vx\in \F^d} \| \abs{A\vx}-\vy\|_2.
\end{equation}
If  $\vx_0$ is sparse, both the constrained nonlinear Lasso model
\begin{equation} \label{eq:mod3}
\min_{\vx\in \F^d} \| \abs{A\vx}-\vy\|_2 \quad \mathrm{s.t.}\quad \|\vx\|_1\le R,
\end{equation}
and its non-constrained  version
\begin{equation}\label{eq:model2}
\min_{\vx\in \F^d} \| \abs{A\vx}-\vy\|_2^2+\lambda \|\vx\|_1,
\end{equation}
have been considered for  recovering $\vx_0$. As we will see later, one already develops
many efficient algorithms to solve (\ref{eq:mod1}). The aim of this paper is to study the performance of (\ref{eq:mod1}) as well as of (\ref{eq:mod3})  and (\ref{eq:model2}) from the theoretical viewpoint.
Particularly, we focus on the question: {\em how well can one  recover $\vx_0$ by solving these above three models?}

\subsection{Algorithms for phase retrieval }

One of the oldest   algorithms for phase retrieval is the error-reduction algorithm which is raised
in \cite{gerchberg1972practical,ER3}. The error-reduction algorithm is to solve the following model
\begin{equation}\label{eq:ERM}
\min_{\vx\in \F^d, C\in \F^{m\times m}} \| A\vx-C\vy\|_2,
\end{equation}
where $C={\rm diag}(c_1,\ldots,c_m)$ with $\abs{c_j}=1, j=1,\ldots,m$.
The error-reduction
is an alternating projection algorithm that iterates between $C$ and $\vx$.
A simple observation  is that  $\vx^\#$ is a solution to (\ref{eq:mod1}) if and only if
$(\vx^\#, {\rm diag}({\rm sign}(A\vx^\#)))$ is a solution to (\ref{eq:ERM}). Hence,  the error-reduction algorithm can be used to solve  (\ref{eq:mod1}). The convergence property of the error-reduction algorithm is studied  in \cite{AltMin,ER4}.
Beyond the error-reduction algorithm, one also develops the generalized gradient descent method for solving (\ref{eq:mod1}) (see \cite{TAF} and \cite{zhang2016reshaped}).

An alternative model for phase retrieval is
\begin{equation}\label{eq:mod2}
\min_{\vx\in \F^d}\,\,\sum_{i=1}^m \xkh{\abs{\nj{\va_i,\vx}}^2-y_i^2}^2.
\end{equation}
Although the objective function in (\ref{eq:mod2}) is non-convex, many computational algorithms turn to be successful actually with a good initialization, such as Gauss-Newton algorithms \cite{Gaoxu}, Kaczmarz algorithms  \cite{tan2017phase}  and trust-region methods \cite{turstregion}. A gradient descent method is applied to solve (\ref{eq:mod2}), which provides  the Wirtinger Flow (WF) \cite{WF} and Truncated Wirtinger Flow (TWF) \cite{TWF}  algorithms. It has been proved that both WF and TWF algorithms linearly converge to the true solution up to a global phase.  For the sparse phase retrieval, a standard $\ell_1$ norm term is added to the above objective functions to obtain the models for sparse phase retrieval, such as  (\ref{eq:mod3}) and (\ref{eq:model2}). Similarly, the gradient descent method with thresholding can be used to solve those models successfully \cite{cai2016optimal,SparseTAF}.

 One convex method to handle phase retrieval problem is PhaseLift \cite{phaselift} which lifts the quadratic system to recover a  rank-1 positive semi-definite matrix by solving a semi-definite programming. An alternative convex method is PhaseMax \cite{phasemax} which recasts this problem as a linear programming by  an anchor vector.


\subsection{Our contributions} \label{contribution}
The aim of this paper is to study the estimation  performance of the nonlinear least squares for phase retrieval.  We obtain the measurement vector $\vy=|A\vx_0|+\eta$, where $A=[\va_1,\ldots,\va_m]^\T$ is the measurement matrix with $\va_j \in \Rd, \vx_0\in \Rd$ and $\eta \in \R^m$ is a noise vector. We would like to estimate $\vx_0$ from $\vy$.

Firstly, we consider  the following non-linear least square model:
\begin{equation} \label{square model}
  \mathop{\min} \limits_{\vx \in \Rd} \quad \left\||A\vx|-\vy\right\|^2.
\end{equation}
One of main results is the following theorem which shows that the reconstruction  error of model (\ref{square model}) can be reduced proportionally to $\norm{\ve}/\sqrt{m}$ and it becomes quite small when $\|\eta\|_2$ is bounded and $m$ is large.

\begin{theorem} \label{result1}
Suppose that $A\in \R^{m\times d}$ is a Gaussian random matrix whose entries are independent Gaussian random variables.
We assume that  $m \gtrsim d$. The following holds with probability at least $1-3\exp(-c m)$. For any fixed vector $\vx_0\in \Rd$, suppose that  $\x \in \Rd$ is any solution  to  (\ref{square model}). Then \begin{equation}\label{eq:orderm}
  \min \left\{\norm{\x-\vx_0}, \norm{\x+\vx_0}\right\} \lesssim  \frac{\|\eta\|_2}{\sqrt{m}}.
\end{equation}
\end{theorem}

The next theorem implies that the reconstruction error in Theorem \ref{result1} is sharp.

\begin{theorem}\label{th:con_lower}
Let $m \gtrsim d$.
Assume that  $\vx_0\in \Rd$  is a fixed vector. Assume that
$\eta\in \R^m$ is a fixed vector which satisfies $\sqrt{2/\pi}\cdot\abs{\sum_{i=1}^m \eta_i}/m\ge \delta_0$ and $\norm{\eta}/\sqrt{m} \le \delta_1$ for some $\delta_0>0$ and $\delta_1> 0$.
 Suppose that $A\in \R^{m\times d}$ is a Gaussian random matrix whose entries are independent Gaussian random variables.  Let $\x$ be any solution to (\ref{square model}). Then there exists a $\epsilon_0>0$ and a constant $c_{\delta_0,\vx_0}>0$ such that the following holds with probability at least $1-6\exp(-c\epsilon_0^2 m)$:
\begin{equation}\label{eq:constant_lower}
  \min \left\{\norm{\x-\vx_0}, \norm{\x+\vx_0}\right\} \ge c_{\delta_0,\vx_0}.
\end{equation}
Here, the constant $c_{\delta_0,\vx_0}$ only depends on $\delta_0$ and $\norm{\vx_0}$.
\end{theorem}

\begin{remark} \label{lowbound_optimal}
We next explain the reason why the error bound in Theorem \ref{result1} is sharp up to a constant.
For the aim of contradiction, we assume that there  exists a $\alpha>0$ such that
\begin{equation}\label{eq:ascon}
  \min \left\{\norm{\x-\vx_0}, \norm{\x+\vx_0}\right\} \lesssim \frac{\norm{\eta}}{m^{1/2+\alpha}}\quad \text{ for } m\gtrsim d,
\end{equation}
holds   for any fixed  $\vx_0\in {\mathbb R}^d$ with high probability. Here,   $\x \in \Rd$  is any solution to  (\ref{square model}) which depends on $\vx_0$ and $\eta$.
We assume
\begin{equation*}
 \varliminf _{m\to \infty }\abs{\sum_{i=1}^m \eta_i/m}\ge \delta_0 \quad \mathrm{and} \quad \varlimsup _{m\to \infty } \norm{\eta}/\sqrt{m} \le \delta_1
\end{equation*}
where $\delta_0,\delta_1>0$. For example, if we  take $\eta=(1,\ldots,1)^T\in {\mathbb R}^m$, then $\delta_0=\delta_1=1$. For a fixed $\vx_0\in \R^d$, Theorem \ref{th:con_lower} implies the following  holds
with high probability
\begin{equation}\label{eq:remark_lower}
  \min \left\{\norm{\x-\vx_0}, \norm{\x+\vx_0}\right\} \ge c_{\delta_0,\vx_0}, \quad \text{ for } m\gtrsim d,
\end{equation}
where  $c_{\delta_0,\vx_0}>0$.
However, the (\ref{eq:ascon}) implies that
\begin{equation*}
  \min \left\{\norm{\x-\vx_0}, \norm{\x+\vx_0}\right\}\lesssim \frac{\delta_1}{m^\alpha} \to 0, \quad m\to \infty,
\end{equation*}
which contradicts with (\ref{eq:remark_lower}). Hence, (\ref{eq:ascon}) does not hold.
\end{remark}
We next turn to the phase retrieval for sparse signals. Here, we assume that $\vx_0\in \R^d$ is $s$-sparse, which means that there are at most $s$ nonzero entries in $\vx_0$.
 We first consider the estimation performance of the following constrained nonlinear Lasso model
\begin{equation} \label{square model with sparse with constrain}
\min_{\vx\in \R^d} \| \abs{A\vx}-\vy\|_2 \quad \mathrm{s.t.}\quad \|\vx\|_1\le R,
\end{equation}
where $R$ is a parameter which specifies a desired sparsity level of the solution.   The following theorem presents the estimation performance  of model (\ref{square model with sparse with constrain}):
\begin{theorem} \label{th:sparse constrain}
Suppose  that $A\in \R^{m\times d}$ is a Gaussian random matrix whose entries are independent Gaussian random variables. If $m \gtrsim s\log(ed/s)$, then the following holds with probability at least $1-3\exp(-c_0 m)$ where $c_0>0$ is a constant. For any fixed $s$-sparse vector $\vx_0\in \Rd$,  suppose that  $\x \in \Rd$ is any solution to  (\ref{square model with sparse with constrain}) with parameter $R:=\norms{\vx_0}_1$ and $\vy=\abs{A\vx_0}+\eta$. Then
\begin{equation*}
  \min \left\{\norm{\x-\vx_0}, \norm{\x+\vx_0}\right\}\,\, \lesssim\,\,  \frac{\norm{\omega}}{\sqrt{m}}.
\end{equation*}
\end{theorem}
The  unconstrained Lagrangian version of (\ref{square model with sparse with constrain}) is
 \begin{equation}\label{square model with sparse}
   \mathop{\min} \limits_{\vx \in \Rd} \quad \left\||A\vx|-\vy\right\|^2+\lambda\norms{\vx}_1,
 \end{equation}
 where $\lambda>0$ is a parameter which depends on the desired level of sparsity.
The following theorem  presents  the estimation  performance of model (\ref{square model with sparse}):
\begin{theorem} \label{th:sparse}
Suppose that $A\in \R^{m\times d}$ is a Gaussian random matrix whose entries are independent Gaussian random variables. If $m\gtrsim s\log (ed/s)$, then the following holds with probability at least $1-\exp(-c_0 m)-1/d^2$ where $c_0>0$ is a constant.
For any fixed $s$-sparse vector $\vx_0\in \Rd$,  suppose that  $\x \in \Rd$ is any solution to
(\ref{square model with sparse}) with the positive parameter $\lambda \gtrsim  \normone{\eta}+\norm{\eta}\sqrt{\log d}$ and  $\vy=\abs{A\vx_0}+\eta$.
Then
\begin{equation}\label{eq:nonlam}
  \min \left\{\norm{\x-\vx_0}, \norm{\x+\vx_0}\right\} \lesssim \frac{\lambda\sqrt{s}}{m}+\frac{\norm{\omega}}{\sqrt{m}}.
\end{equation}
\end{theorem}

We can use a similar method to that in Remark (\ref{lowbound_optimal}) to show that the reconstruction error in Theorem \ref{th:sparse constrain} is sharp. In Theorem \ref{th:sparse}, one requires that $\lambda \gtrsim  \normone{\eta}+\norm{\eta}\sqrt{\log d}$.
Motivated by a lot of numerical experiments, we conjecture that Theorem \ref{th:sparse}  still holds provided  $\lambda \gtrsim  \norm{\eta}\sqrt{\log d}$. If the conjecture holds, then we can take $\lambda \thickapprox \norm{\eta}\sqrt{\log d}$ and replace (\ref{eq:nonlam}) by
\[
  \min \left\{\norm{\x-\vx_0}, \norm{\x+\vx_0}\right\} \lesssim \frac{\norm{\omega}}{\sqrt{m}}.
\]

\subsection{Comparison to related works }
\subsubsection{Least squares}
We first introduce the estimation  of signals from the noisy linear measurements.
Suppose that $\vx_0\in \R^d$ is the target signals. Set
\begin{equation*}
  \vy'=A\vx_0+\ve,
\end{equation*}
where $A \in \R^{m\times d}$ is the measurement matrix and $\ve \in \R^m$ is a noise vector.
We suppose that $A$ is a Gaussian random matrix with entries $a_{jk}\sim N(0,1)$ and we also
suppose that $m\gtrsim d$.
A popular method for recovering $\vx_0$ from $\vy'$ is the least squares:
\begin{equation} \label{linear lasso without sparse}
  \min_{\vx\in \Rd} \norm{A\vx-\vy'}^2.
\end{equation}
Then the solution of model (\ref{linear lasso without sparse}) is $\xx=(A^\T A)^{-1}A^\T \hy$, which implies that
\begin{equation*}
\xx-\vx_0=(A^\T A)^{-1}A^\T \eta.
\end{equation*}
Thus with probability at least $1-4\exp(-c d)$ one has
\begin{equation*}
\norm{\xx-\vx_0}=\norm{(A^\T A)^{-1}A^\T \eta}\le \norm{(A^\T A)^{-1}}\norm{A^\T \eta}\lesssim\frac{\sqrt{d}}{m}\norm{\eta},
\end{equation*}
where the last inequality follows from the fact that $\norm{A^\T \eta}\le 3\sqrt{d}\norm{\eta}$ and $\lambda_{\min}(A) \ge O(\sqrt{m})$ hold with probability at least $1-4\exp(-c d)$ for any Gaussian random matrix \cite[Theorem 7.3.3]{Vershynin2018}.
Then the following holds with high probability
\begin{equation}\label{eq:error}
 \norm{\xx-\vx_0} \lesssim \frac{\sqrt{d}\|\eta\|_2}{m},
\end{equation}
where $\xx$ is the solution of (\ref{linear lasso without sparse}).

For non-linear least squares with phaseless measurement $\vy=|A \vx_0|+\eta$, we consider
\begin{equation} \label{mol:nonlinear}
\min_{\vx \in \Rd} \|\abs{A\vx}-\vy\|.
\end{equation}
Theorem \ref{result1} implies that
 \begin{equation}\label{eq:plerror}
   \min \left\{\norm{\x-\vx_0}, \norm{\x+\vx_0}\right\}\,\, \lesssim\,\, \frac{\|\eta\|_2}{\sqrt{m}}
 \end{equation}
where $\x$ is any solution to (\ref{mol:nonlinear}).
Remark \ref{lowbound_optimal} implies that the upper bound is sharp.
Note that the error order about $m$ for nonlinear least squares is $O(1/\sqrt{m})$
while  one for  least squares is $O(1/m)$.
Hence, the result in  Theorem \ref{result1}
highlights an essential difference between linear least square model (\ref{linear lasso without sparse}) and  the non-linear least square model (\ref{mol:nonlinear}).

\subsubsection{Lasso}

If  assume that the signal $\vx_0$ is $s$-sparse and   $\vy'=A\vx_0+\ve$, one turns to the Lasso
\begin{equation} \label{lasso:sparse uncons}
\min_{\vx\in \R^d} \| A\vx-\vy'\|_2 \quad \mathrm{s.t.}\quad \|\vx\|_1\le R.
\end{equation}
If $m \gtrsim s\log d$, then  the solution $\xx$ of (\ref{lasso:sparse uncons}) satisfies
\begin{equation}\label{eq:lassoer}
  \norm{\xx-\vx_0} \lesssim \|\eta\|_2 \sqrt{s\log d}/m
\end{equation}
with high probability (see \cite{Vershynin2018}).

For the nonlinear Lasso, Theorem \ref{th:sparse constrain} shows that any solution $\x$ to  $\min_{\normone{\vx} \le \normone{\vx_0}} \|\abs{A\vx}-\vy\|$ with $\vy=|A \vx_0|+\eta$ satisfies
 \begin{equation}\label{eq:nlassoer}
   \min \left\{\norm{\x-\vx_0}, \norm{\x+\vx_0}\right\} \lesssim \norm{\eta}/\sqrt{m}
 \end{equation}
with high probability.
Comparing (\ref{eq:lassoer}) with (\ref{eq:nlassoer}), we find that the reconstruction error of Lasso is similar
to that of nonlinear Lasso when $m=O(s\log d)$, while Lasso has the better performance
over the nonlinear Lasso provided  $m\gg s\log d$.

\subsubsection{Unconstrained  Lasso}

We next turn to the unconstrained  Lasso
\begin{equation}\label{model:uncons}
  \min_{\vx \in \Rd} \|A\vx-\vy'\|^2+ \lambda \normone{\vx}
\end{equation}
where $\vy'=A \vx_0+\eta$ and $\vx_0$ is a $s$-sparse vector. If the parameter $\lambda \gtrsim \norm{\eta}\sqrt{\log d}$, then  $\xx$ satisfies
 \[
\norm{\xx-\vx_0}\lesssim \frac{ \lambda \sqrt{s}}{m}
\]
 with high probability (see \cite{Vershynin2018}) where $\xx$ is the solution of (\ref{model:uncons}).

 For  the sparse phase retrieval model
\begin{equation}\label{eq:sparselasso}
\min_{\vx \in \Rd} \|\abs{A\vx}-\vy\|^2+ \lambda \normone{\vx}
\end{equation}
 with $\vy=|A \vx_0|+\eta$,  Theorem \ref{th:sparse} shows that
 \begin{equation}\label{eq:sparselamda}
  \min \left\{\norm{\x-\vx_0}, \norm{\x+\vx_0}\right\} \lesssim \frac{\lambda\sqrt{s}}{m}+\frac{\norm{\omega}}{\sqrt{m}}
\end{equation}
where the parameter $\lambda \gtrsim  \normone{\eta}+\norm{\eta}\sqrt{\log d}$
and $\x$ is any solution to (\ref{eq:sparselasso}).
Our  result requires that the parameter $\lambda$ in nonlinear Lasso model is larger than linear case.

\subsubsection{The generalized Lasso with nonlinear observations }
 In \cite{plan2016generalized}, Y. Plan and R. Vershynin consider the following non-linear observations
\begin{equation*}
  y_j=f_j(\nj{\va_j,\vx_0}),\quad j=1,\ldots,m
\end{equation*}
where $f_j: \R\to \R$ are independent copies of an unknown random or deterministic function $f$ and $\va_j\in \Rd, j=1,\ldots,m,$ are Gaussian random vectors. The $K$-Lasso model is employed to recover $\vx_0$ from $y_j, j=1,\ldots,m$:
\begin{equation}\label{eq:nonlinear}
  \min_{\vx\in \Rd} \norm{A\vx-\vy}^2  \quad \mathrm{s.t.}\quad \vx\in K,
\end{equation}
where $K\subset \Rd$ is some known set. Suppose that $\x$ is the solution to (\ref{eq:nonlinear}). Y. Plan and R. Vershynin  \cite{plan2016generalized} show that $\|\x-\mu\cdot \vx_0\|$ tends to 0 with $m$ tending to infinity, where $\mu={\mathbb E}(f(g)g)$ with $g$ being a Gaussian random variable.
Unfortunately, applying the  result to phase retrieval problem, it gives that $\mu={\mathbb E}(|g|\cdot g)=0$  and hence $\|\x\|$ tends to 0 with  $m$ tending to infinity where
$\x$ is the solution to the least square mode  (\ref{eq:nonlinear}) with $K=\Rd$ and $y_j=\abs{\nj{\va_j,\vx_0}}$. This means that the generalized Lasso does not work for phase retrieval.
Hence, one has to employ the nonlinear  Lasso (or nonlinear least squares) for solving phase retrieval. This is also our motivation for this project.

\subsection{Organization}
The paper is organized as follows.  In Section 2, we introduce some notations and lemmas which are used in this paper.  We provide the proofs of main results  in Section 3.

\section{Preliminaries}
The aim of this section is to  introduce some definitions and lemmas which play a key role in our paper.
\subsection{Gaussian width}
For a subset $T\subset \R^d$, the Gaussian width is defined as
\begin{equation*}
  w(T):= \E \sup_{\vx\in T} \jkh{g,\vx} \quad \mathrm{where} \quad g \sim N(0,I_d).
\end{equation*}
The Gaussian width $w(T)$ is one of the basic geometric quantities associated with the subset $T\subset \R^d$ (see \cite{Vershynin2018}). We now give several examples about Gaussian width. The first example is Euclidean unit ball $\mathbb{S}^{d-1}$, where a simple calculation leads to
\begin{equation*}
  w(\mathbb{S}^{d-1})=O(\sqrt{d}).
\end{equation*}
Another example is the unit $\ell_1$ ball $B_1^d$ in $\Rd$. It can be showed that (see e.g. \cite{Vershynin2018})
\begin{equation*}
  w(B_1^d)=O(\sqrt{\log d}).
\end{equation*}
In this paper, we often use the following set
\begin{equation*}
K_{d,s}:=\dkh{\vx \in \Rd: \norm{\vx}\le 1, \quad \normone{\vx}\le \sqrt{s}},
\end{equation*}
with the Gaussian width  $w(K_{d,s})=O(\sqrt{s\log(ed/s)})$ (see e.g. \cite{Vershynin2018}).

\subsection{Gaussian Concentration Inequality}
\begin{lemma}\cite{Vershynin2018} \label{Gaussian space}
Consider a random vector $X\sim N(0,I_d)$ and a Lipschitz function $f:\Rd \to \R$ with constant $\norms{f}_{\mathrm{Lip}}$: $\abs{f(X)-f(Y)}\le \norms{f}_{\mathrm{Lip}}\cdot \norm{X-Y}$. Then for every $t\ge 0$, we have
\begin{equation*}
\PP\dkh{\abs{f(X)-\E f(X)}\ge t}\le 2\exp\left(-\frac{ct^2}{\norms{f}_{\mathrm{Lip}}}\right).
\end{equation*}

\end{lemma}

\subsection{Strong RIP}
To study the phaseless compressed sensing,  Voroninski and Xu  introduce  the definition of strong restricted isometry property (SRIP) (see \cite{Voroninski2016A}).
\begin{definition}\cite{Voroninski2016A}
The matrix $A\in \R^{m\times d}$ satisfies the Strong Restricted Isometry Property of order $s$ and constants $\theta_{-}, \theta_{+}\in (0,2)$ if the following holds
\begin{equation}\label{eq:strongrip}
  \theta_{-}\norm{\vx}^2 \le \min\limits_{I\subset [m],\abs{I}\ge m/2} \norm{A_I \vx}^2\le \max\limits_{I\subset [m],\abs{I}\ge m/2} \norm{A_I \vx}^2 \le \theta_{+}\norm{\vx}^2
\end{equation}
for all  $\vx\in K_{d,s}$. Here, $A_I$ denotes the submatrix of $A$ where only {\em rows} with indices in $I$ are kept, $[m]:=\{1,\ldots,m\}$ and $\abs{I}$ denotes the cardinality of $I$.
\end{definition}

The following lemma shows that Gaussian random matrices satisfy SRIP with high probability for some non-zero universal constants $\theta_{-}, \theta_{+}>0$.
\begin{lemma}\cite[Theorem 2.1]{Voroninski2016A} \label{SRIP}
Suppose that $t>1$ and that $A\in \R^{m\times d}$ is a Gaussian random matrix with entries $a_{jk}\sim N(0,1)$. Let $m=O(tk\log(e d/k))$. Then there exist constants $\theta_{-}, \theta_{+}$ with $0<\theta_{-}< \theta_{+}<2$, independent with $t$, such that $A/\sqrt{m}$ satisfies SRIP of order $t\cdot k$ and constants $\theta_{-}, \theta_{+}$ with probability at least $1-\exp(-cm/2)$, where $c>0$ is an absolute constant.
\end{lemma}
\begin{remark}\label{re:SRIP}
In \cite{Voroninski2016A}, the authors just present the proof of Lemma \ref{SRIP} for the case where $\vx$ is $s$-sparse.
Note that the set $K_{d,s}$ has covering number $N(K_{d,s},\varepsilon)\le \exp(Cs \log (ed/s)/\varepsilon^2)$ \cite[Lemma 3.4]{Plan2013}. It is easy to extend the proof in \cite{Voroninski2016A} to the case where $\vx \in K_{d,s}$.
\end{remark}

\section{Proof of the main results}
\subsection{Proof of Theorem \ref{result1}}
We begin with a simple lemma.
\begin{lemma} \label{upper bound}
Suppose that $m\ge d$. Let $A\in \R^{m\times d}$ be a Gaussian matrix whose entries are independent Gaussian random variables. Then the following holds with probability at least $1-2\exp(-cm)$
\begin{equation*}
  \sup_{\vh\in \Rd \atop \eta \in \R^m}\nj{\vh,A^\T \eta}\le 3\sqrt{m}\norm{\vh}\norm{\eta}.
\end{equation*}

\end{lemma}
\begin{proof}
Since $A\in \R^{m\times d}$ is a Gaussian random matrix, we have $\norm{A} \le 3\sqrt{m}$ with probability at least $1-2\exp(-c m)$ \cite[Theorem 7.3.3]{Vershynin2018}.
We obtain that
\begin{equation*}
  \nj{\vh,A^\T \eta} \le \norm{\vh} \norm{A^\T\eta} \le \norm{\vh} \norm{A^\T} \norm{\eta}\le 3\sqrt{m}\norm{\vh}\norm{\eta}
\end{equation*}
holds with probability at least $1-2\exp(-c m)$. We arrive at the conclusion.
\end{proof}

\begin{proof}[Proof of Theorem \ref{result1}]
Set $\vh^{-}:=\x-\vx_0$ and $\vh^{+}:=\x+\vx_0$. Since $\x$ is the solution of (\ref{square model}), we have
\begin{equation} \label{minsolution}
  \left\||A\x|-\vy\right\|^2\le \left\||A\vx_0|-\vy\right\|^2.
\end{equation}
For any index set $T\subset \{1,\ldots,m\}$, we let $A_T:=[\va_j:\;j\in T]^\T$ be the submatrix of $A$. Denote
\begin{eqnarray*}
  && T_1:=\left\{j: \mathrm{sign}(\nj{\va_j,\x})=1,\; \mathrm{sign}(\nj{\va_j,\vx_0})=1\right\}  \\
  && T_2:=\left\{j: \mathrm{sign}(\nj{\va_j,\x})=-1,\; \mathrm{sign}(\nj{\va_j,\vx_0})=-1\right\}  \\
  && T_3:=\left\{j: \mathrm{sign}(\nj{\va_j,\x})=1,\; \mathrm{sign}(\nj{\va_j,\vx_0})=-1\right\} \\
  && T_4:=\left\{j: \mathrm{sign}(\nj{\va_j,\x})=-1,\; \mathrm{sign}(\nj{\va_j,\vx_0})=1\right\}.
\end{eqnarray*}
Without loss of generality, we assume that $\# (T_{1}\cup T_2)=\beta m \ge m/2$ (otherwise, we can assume that $\# (T_{3}\cup T_4) \ge m/2$~).
Then we have
\begin{equation*}
 \left\||A\x|-\vy\right\|^2\geq \norm{A_{T_1}\vh^{-}-\eta_{T_1}}^2+\norm{A_{T_2}\vh^{-}+\eta_{T_2}}^2.
\end{equation*}
The (\ref{minsolution}) implies that
\[
\norm{A_{T_1}\vh^{-}-\eta_{T_1}}^2+\norm{A_{T_2}\vh^{-}+\eta_{T_2}}^2\leq \|\eta\|^2
\]
and hence
\begin{equation}\label{key eq}
  \norm{A_{T_{12}}\vh^{-}}^2 \le 2\nj{\vh^{-},A_{T_1}^\T \eta_{T_1}-A_{T_2}^\T \eta_{T_2}}+ \|\eta_{T_{12}^c}\|^2
\end{equation}
where $T_{12}:=T_1 \cup T_2$.   Lemma \ref{SRIP} implies that
\begin{equation}\label{lower bound}
  \norm{A_{T_{12}}\vh^{-}}^2 \ge cm \norm{\vh^{-}}^2
\end{equation}
holds with probability at least $1-\exp(-c_0 m)$. On the other hand, Lemma \ref{upper bound} states that with  probability at least $1-2\exp(-cm)$ the following holds:
\begin{eqnarray}\label{first term}
  \nj{\vh^{-},A_{T_1}^\T \eta_{T_1}-A_{T_2}^\T \eta_{T_2}} &\le & 6\sqrt{m}\norm{\vh^{-}}\norm{\eta}.
\end{eqnarray}
Putting (\ref{lower bound}) and (\ref{first term})  into (\ref{key eq}), we obtain
\begin{equation}\label{eq:bu}
  cm \norm{\vh^{-}}^2\le 12\sqrt{m}\norm{\vh^{-}}\norm{\eta} + \|\eta_{T_{12}^c}\|_2^2
\end{equation}
 with probability at least $1-3\exp(-c_1 m)$, which implies that
\begin{equation*}
  \norm{\vh^{-}} \lesssim \frac{\norm{\eta}}{\sqrt{m}}.
\end{equation*}
For the case where $\# (T_3\cup T_4)\geq m/2$, we can obtain that
\begin{equation*}
  \norm{\vh^{+}} \lesssim \frac{\norm{\eta}}{\sqrt{m}}
\end{equation*}
by a similar method to above.

\end{proof}

\subsection{Proof of Theorem \ref{th:con_lower}}
 To this end,  we present the following lemmas.

\begin{lemma} \label{fixed-point condition}
Suppose that $\x$ is any solution of model (\ref{square model}). Then $\x$ satisfies the following fixed-point equation:
\begin{equation} \label{fixed point}
  \x=(A^\T A)^{-1}A^\T (\vy\odot\s(A\x)),
\end{equation}
where $\odot$ denotes the Hadamard product and $\s(A\x):=
\left(\frac{\nj{\va_1,\x}}{\abs{\nj{\va_1,\x}}},\ldots,\frac{\nj{\va_m,\x}}{\abs{\nj{\va_m,\x}}}\right)$ for any $\x\in\Rd$. Here, $\frac{\nj{\va_j,\x}}{\abs{\nj{\va_j,\x}}}=1$ is adopted if $\nj{\va_j,\x}=0$.
\end{lemma}
\begin{proof}
Let
\[
L(\vx)\,\,:=\,\,\norm{\abs{A \vx}-\vy}^2.
\]
Consider the smooth function
\begin{equation*}
  G(\vx,\vu)\,\,:=\,\,\norm{A \vx-\vu\odot\vy}^2
\end{equation*}
with $\vx\in \Rd$ and $\vu\in U:=\{\vu=(u_1,\ldots,u_m)\in \R^m: \abs{u_i}=1,\; i=1,\ldots,m\}$. Recall that $L(\vx)$ has a global minimum  at $\x$.  Then $G(\vx,\vu)$ has a global
minimum  at $(\x,s(A\x))$. Indeed, if there exists $(\widetilde{\vx},\widetilde{\vu})$ such that $G(\widetilde{\vx},\widetilde{\vu})<G(\x,s(A\x))$, then
\begin{equation*}
  L(\widetilde{\vx})=\norm{\abs{A \widetilde{\vx}}-\vy}^2\le \norm{A \widetilde{\vx}-\widetilde{\vu}\odot\vy}^2=G(\widetilde{\vx},\widetilde{\vu})<G(\x,s(A\x))=L(\x).
\end{equation*}
This contradicts the assumption that $L(\vx)$ has a global minimum at $\x$.
Thus we have
\begin{equation*}
 G(\x,s(A\x)) \le G(\vx,s(A\x)) \quad \text{for  any} \quad \vx \in \Rd,
\end{equation*}
i.e.,  the  function $G(\vx,\s(A\x))$ has a global minimum at $\x$. Here, we consider
  $G(\vx,\s(A\x))$ as a function about $\vx$ since $\s(A\x)$ is a fixed vector. Note that $G(\vx,\s(A\x))$ is differentiable and
\begin{equation*}
  \nabla G(\vx,\s(A\x))=2A^\T(A \vx-\vy\odot\s(A\x)).
\end{equation*}
And $G(\vx,\s(A\x))$ has a global minimum at $\x$, we have
\[
\nabla G(\x,\s(A\x))=2A^\T(A \x-\vy\odot\s(A\x))=0
\]
 which implies the conclusion.
\end{proof}
\begin{lemma} \label{Th:lower bound}
Let $m \gtrsim d$. Suppose that $A\in \R^{m\times d}$ is a Gaussian random matrix whose entries are independent Gaussian random variables. For a fixed vector $\vx_0\in \Rd$ and a fixed noise vector $\eta\in \R^m$, let $\x$ be the solution of model (\ref{square model}). For any fixed $\epsilon>0$,  set
\[
\beta_\epsilon:=\left|\norm{\vx_0}\cdot f(\theta)+\sqrt{2/\pi}\cdot \sum_{i=1}^m \eta_i/m \right|-(\norm{\vx_0}+\norm{\eta}/\sqrt{m})\epsilon,
\]
 where $f(\theta):=2/\pi\cdot (\sin \theta +(\pi/2-\theta)\cos \theta)-|\cos\theta|$ and $\theta$ is the angle between $\x$ and $\vx_0$. Then the following holds with probability at least $1-6\exp(-c\epsilon^2 m)$:
\begin{equation*}
  \min \left\{\norm{\x-\vx_0}, \norm{\x+\vx_0}\right\} \ge \beta_\epsilon/9.
\end{equation*}
\end{lemma}
\begin{proof}
 According to Lemma \ref{fixed-point condition}, we have
\begin{equation} \label{eq:fixed point}
  \x=(A^\T A)^{-1}A^\T (\vy\odot\s(A\x)).
\end{equation}
Without loss of generality, we can assume $\norm{\x-\vx_0}\le \norm{\x+\vx_0}$, which implies that $0\le \theta\le \pi/2$. From (\ref{eq:fixed point}), we have
\begin{equation*}
  \x-\vx_0=(A^\T A)^{-1}A^\T (\vy\odot\s(A\x)-A\vx_0),
\end{equation*}
which implies that
\begin{equation*}
  \norm{\x-\vx_0}\ge \sigma_{\min} ((A^\T A)^{-1})\norm{A^\T (\vy\odot\s(A\x)-A\vx_0)}\ge \frac{1}{9m} \norm{A^\T (\vy\odot\s(A\x)-A\vx_0)}.
\end{equation*}
Here, we use the fact that $\norm{A} \le 3\sqrt{m}$ holds with probability at least $1-2\exp(-c m)$ \cite[Theorem 7.3.3]{Vershynin2018}  since $A\in \R^{m\times d}$ is a Gaussian random matrix.

Without loss of generality, we can assume $\x\neq 0$. Indeed,  (\ref{eq:fixed point}) implies  $A^\T \vy=0$ provided   $\x= 0$, which gives that $\vx_0=0$ and $\eta=0$. Thus our conclusion holds.
 By the unitary invariance of Gaussian random vectors, we can take $\x=\norm{\x}\vve_1$ and $\vx_0=\norm{\vx_0}(\cos \theta\cdot \vve_1+\sin \theta \cdot \vve_2)$, where $\theta$ is the angle between $\x$ and $\vx_0$. Thus,
\begin{equation*}
  \norm{\x-\vx_0}\ge \frac{1}{9m} \norm{A^\T (\vy\odot\s(A\mathbf{e}_1)-A\vx_0)}=\frac{1}{9m} \norm{\vz},
\end{equation*}
where $\vz:=(z_1,\ldots,z_d)^\T:=A^\T (\vy\odot\s(A\mathbf{e}_1)-A\vx_0)$.
Note that the first entry of $\vz$ is
\begin{equation*}
  z_1=\sum_{i=1}^m \xkh{\abs{a_{i,1}}(\abs{a_i^\T \vx_0}+\eta_i)-a_{i,1}\cdot a_i^\T \vx_0}.
\end{equation*}
This implies that
\begin{equation}\label{two part sum}
\begin{aligned}
  \norm{\x-\vx_0} \ge    \frac{\abs{z_1}}{9m}= &\left|\norm{\vx_0}\cdot\frac{1}{9m}\sum_{i=1}^m \big|a_{i,1}(a_{i,1}\cos\theta+a_{i,2}\sin\theta)\big|+\frac{1}{9m}\sum_{i=1}^m\eta_i\abs{a_{i,1}} \right. \\
   & \left. - \norm{\vx_0}\cdot\frac{1}{9m}\sum_{i=1}^m a_{i,1}(a_{i,1}\cos\theta+a_{i,2}\sin\theta)\right| \\
=\,\,&\left|\frac{\norm{\vx_0}}{9m}\sum_{i=1}^m(\abs{\xi_i}-\xi_i)
+\frac{1}{9m}\sum_{i=1}^m\eta_i\abs{a_{i,1}}\right|,
\end{aligned}
\end{equation}
 where $\xi_i:=a_{i,1}(a_{i,1}\cos\theta+a_{i,2}\sin\theta)$. It is clear that $\xi_i$ is a subexponential random variable with $\E \xi_i =\cos \theta$. We claim that  $\E |\xi_i|= 2/\pi\cdot (\sin \theta +(\pi/2-\theta)\cos \theta)$. Then the Bernstein's inequality implies  that, for any fixed $\epsilon>0$,
\begin{equation} \label{The first part}
  \left|\frac{1}{m}\sum_{i=1}^m (|\xi_i|-\xi_i)- \frac{2}{\pi}\cdot (\sin \theta +(\frac{\pi}{2}-\theta)\cos \theta)+\cos\theta \right| \le \epsilon
\end{equation}
holds with probability at least $1-2\exp(-c\epsilon^2 m)$. We next consider
$ \frac{1}{m}\sum_{i=1}^m\eta_i\abs{a_{i,1}}$.
 Note that $\E \abs{a_{i,1}}=\sqrt{2/\pi}$.
Then by Hoeffding's inequality we can obtain that
\begin{equation} \label{The second part}
  \left|\frac{1}{m}\sum_{i=1}^m\eta_i\abs{a_{i,1}}-\sqrt{\frac{2}{\pi}}\cdot \frac{1}{m}\sum_{i=1}^m\eta_i\right| \le \frac{\norm{\eta}}{\sqrt{m}}\epsilon
\end{equation}
holds with probability at least $1-2\exp(-c\epsilon^2 m)$ for any $\epsilon>0$. Substituting (\ref{The first part}) and (\ref{The second part}) into (\ref{two part sum}), we obtain that
\begin{equation} \label{result}
   \norm{\x-\vx_0}\ge \frac{1}{9}\cdot \left(\left|\norm{\vx_0}f(\theta)+\sqrt{\frac{2}{\pi}}\cdot \frac{1}{m}\sum_{i=1}^m\eta_i\right|-
\left(\norm{\vx_0}+\frac{\norm{\eta}}{\sqrt{m}}\right)\epsilon\right)
\end{equation}
holds with probability at least $1-6\exp(-c\epsilon^2 m)$. Thus we arrive at the conclusion.

 It remains to argue that $\E |\xi_i|= 2/\pi\cdot (\sin \theta +(\pi/2-\theta)\cos \theta)$.
 By spherical coordinates integral,
\begin{eqnarray*}
 \E |\xi_i|= \E \big|a_{i,1}(a_{i,1} \cos \theta +a_{i,2} \sin \theta) \big| &=& \frac{1}{2\pi}\int_0^{2\pi}\int_0^\infty  r^3 e^{-r^2/2} |\cos \phi \cos(\theta-\phi)| dr d\phi\\
   &=& \frac{1}{2\pi}\int_0^{2\pi} |\cos \theta +  \cos(2\phi-\theta)|d\phi\\
   &=& \frac{1}{\pi}\int_0^{\pi} |\cos \theta +  \cos \phi|d\phi \\
   &=&  \frac{2}{\pi} \left(\sin \theta +(\pi/2-\theta)\cos \theta\right)
\end{eqnarray*}
where we use the identities $2\cos \phi \cos (\theta-\phi)=\cos \theta +\cos(2\phi-\theta)$ in second line.
\end{proof}

\begin{proof}[Proof of Theorem \ref{th:con_lower}]
From Lemma \ref{Th:lower bound}, it is easy to prove that (\ref{eq:constant_lower}) holds for $\vx_0=0$.
Then it suffices to prove the theorem for $\vx_0\neq 0$.
Since $\norm{\eta}/\sqrt{m} \le \delta_1$ with $\delta_1\ge 0$,  there exists a $\epsilon_0>0$ so that
$$(\norm{\vx_0}+\norm{\eta}/\sqrt{m})\epsilon_0 \le \delta_0/2.$$
Set
\[
\overline{\eta}:=\sqrt{2/\pi}\cdot \sum_{i=1}^m \eta_i/m,
\]
and
\begin{equation*}
  f(\theta):=2/\pi\cdot (\sin \theta +(\pi/2-\theta)\cos \theta)-|\cos\theta|,\quad 0\le \theta \le \pi.
\end{equation*}
Note that  $f(\theta)$ is a monotonically increasing function for  $\theta \in [0, \pi/2]$.

Choosing $\epsilon=\epsilon_0$ in Lemma \ref{Th:lower bound}, with probability at least $1-6\exp(-c\epsilon_0^2 m)$, we have
\begin{equation} \label{eq:basic lower bound}
  \min \left\{\norm{\x-\vx_0}, \norm{\x+\vx_0}\right\} \ge \big(\big|\norm{\vx_0}\cdot f(\theta_0) + \overline{\eta}\big|-\delta_0/2\big) /9,
\end{equation}
where $\theta_0$ is the angle between $\x$ and $\vx_0$. Without loss of generality, we can assume $ 0\le \theta_0 \le \pi/2$ and hence  $f(\theta_0) \ge f(0)=0$.

Noting $\abs{\overline{\eta}}\geq \delta_0$, we divide the rest of the proof into three cases.

{\em Case 1:} $\overline{\eta}\ge \delta_0$.

In this case, (\ref{eq:basic lower bound}) implies that
\begin{equation*}
  \min \left\{\norm{\x-\vx_0}, \norm{\x+\vx_0}\right\} \ge \big(\overline{\eta}-\delta_0/2\big) /9 \ge \delta_0/18
\end{equation*}
holds with probability at least $1-6\exp(-c\epsilon_0^2 m)$.

{\em Case 2:} $\overline{\eta}\le -\delta_0$ and $\abs{\overline{\eta}}\le \norm{\vx_0}\cdot f(\theta_0)$.

In this case,  we have $f(\theta_0)\ge \delta_0/\norm{\vx_0}$.
Since the function $f(\theta)$ is monotonicity, we have  $\theta_0\ge \theta_1:=f^{-1}(\delta_0/\norm{\vx_0})>0$, which implies that
\begin{equation*}
  \min \left\{\norm{\x-\vx_0}, \norm{\x+\vx_0}\right\} \ge \norm{\vx_0}\sin\theta_1.
\end{equation*}

{\em Case 3:} $\overline{\eta}\le -\delta_0$ and $\abs{\overline{\eta}}> \norm{\vx_0}\cdot f(\theta_0)$.

 We claim that there exists a constant $c_{\delta_0,\vx_0}$ such that the following holds
with probability at least $1-6\exp(-c\epsilon_0^2 m)$
\begin{equation} \label{eq:claim_three_part}
  \min \left\{\norm{\x-\vx_0}, \norm{\x+\vx_0}\right\} \ge c_{\delta_0,\vx_0}
\end{equation}
 where $c_{\delta_0,\vx_0}$ only depends on $\delta_0$ and $\norm{\vx_0}$. Indeed, if $\abs{\overline{\eta}}-\norm{\vx_0}f(\theta_0)\ge 3/4\cdot \abs{\overline{\eta}}$, then (\ref{eq:basic lower bound}) implies
\begin{equation*}
  \min \left\{\norm{\x-\vx_0}, \norm{\x+\vx_0}\right\} \ge \big(\abs{\overline{\eta}}-\norm{\vx_0}f(\theta_0)-\delta_0/2\big) /9\ge \delta_0/36.
\end{equation*}
If $\abs{\overline{\eta}}-\norm{\vx_0}f(\theta_0)< 3/4\cdot \abs{\overline{\eta}}$, then $f(\theta_0)\ge \delta_0/(4\norm{\vx_0})$.
It can also give that
\begin{equation*}
  \min \left\{\norm{\x-\vx_0}, \norm{\x+\vx_0}\right\} \ge \norm{\vx_0}\cdot \sin\theta_2,
\end{equation*}
where  $\theta_2:=f^{-1}(\delta_0/(4\norm{\vx_0}))>0$. Choosing $c_{\delta_0,\vx_0}:=\min\{\delta_0/36,\norm{\vx_0}\sin\theta_2\}$, we arrive at the conclusion.
\end{proof}

\subsection{Proof of Theorem \ref{th:sparse constrain}}
We first extend Lemma \ref{upper bound} to sparse case.
\begin{lemma} \label{upper bound sparse cons}
For any fixed $s>0$, let $m \gtrsim s\log (e d/s)$. Suppose that $A\in \R^{m\times d}$ is a Gaussian matrix whose entries are independent Gaussian random variables. Set
\begin{equation*}
  K_{d,s}\,\,:=\,\,\dkh{\vx \in \Rd: \|\vx\|_2\leq 1, \normone{\vx}\le \sqrt{s}}.
\end{equation*}
 Then for any fixed $\eta \in \R^m$, the following holds with probability at least $1-2\exp(-cm)$
\begin{equation}
  \sup_{\vh\in K_{d,s} \atop T\subset \{1,\ldots,m\}}\nj{\vh,A^\T \eta_{T}} \lesssim \sqrt{m}\cdot \norm{\eta}\cdot \norm{\vh},
\end{equation}
where $\eta_T$ denotes the vector generated by $\eta$ with entries in $T$ are themselves and others are zeros.
\end{lemma}
\begin{proof}
For any fixed $T\subset \{1,\ldots,m\}$,  we have
\begin{equation*}
 \E \sup_{\vh\in K_{d,s}}\nj{\vh,A^\T \eta_{T}} =\norm{\eta_{T}}\cdot w(K_{d,s})\le  C\sqrt{s\log (ed/s)}\norm{\eta}\le C\sqrt{m}\norm{\eta},
\end{equation*}
where the first inequality follows from the fact of the Gaussian width $w(K_{d,s}) \le C\sqrt{s\log (ed/s)}$ and the second inequality follows from $m \ge c_0s\log (ed/s)$.  We next use  Lemma \ref{Gaussian space}  to give a tail bound for $\sup_{\vh\in K_{d,s}} \nj{\vh,A^\T \eta_{T}}$. To this end, we set
\[
f(A):= \sup_{\vh\in K_{d,s}} \nj{\vh,A^\T \eta_{T}}.
\]
We next show that $f(A)$ is a Lipschitz function on $\R^{m\times d}$ and its Lipschitz constant is $\norm{\eta}$. Indeed, for any matrices $A_1,A_2 \in \R^{m\times d}$, it holds that
\begin{equation*}
  \Big|\sup_{\vh\in K_{d,s}} \nj{\vh,A_1^\T \eta_{T}}-\sup_{\vh\in K_{d,s}} \nj{\vh,A_2^\T \eta_{T}}\Big| \le  \Big|\sup_{\vh\in K_{d,s}} \nj{(A_1-A_2)\vh,\eta_{T}}\Big|\le \norm{\eta}\norms{A_1-A_2}_F.
\end{equation*}
Then Lemma \ref{Gaussian space} implies that
\begin{equation}\label{eq:budeng}
  \PP \dkh{\sup_{\vh\in K_{d,s}} \nj{\vh,A^\T \eta_{T}}\ge \E \sup_{\vh\in K_{d,s}} \nj{\vh,A^\T \eta_{T}}+t} \le 2\exp\Big(-\frac{ct^2}{\norm{\eta}^2}\Big).
\end{equation}
Suppose that $C_1>0$ is a constant satisfying $ C_1^2\cdot c>1$.
Choosing $t=C_1\sqrt{m}\norm{\eta}$ in (\ref{eq:budeng}), we obtain that the following holds
with probability at least $1-2\exp(-c \cdot C_1^2\cdot m)$
\begin{equation*}
  \sup_{\vh\in K_{d,s}} \nj{\vh,A^\T \eta_{T}} \le C_0\sqrt{m}\norm{\eta}
\end{equation*}
for any fixed $T\subset \{1,\ldots,m\}$ .

Finally, note that the number of all subset $T\subset \{1,\ldots,m\}$ is $2^m$.  Taking a union bound over all the sets gives
\begin{equation*}
  \sup_{\vh \in K_{d,s} \atop T\subset \{1,\ldots,m\}}\nj{\vh,A^\T \eta_{T}} \le C_0\sqrt{m}\norm{\eta}
\end{equation*}
 with probability at least $1-2\exp(-\tilde{c}m)$. Here, we use the fact of $C_1^2\cdot c>1$. We  arrive at the conclusion.
\end{proof}

\begin{proof}[Proof of Theorem \ref{th:sparse constrain}]
Set $\vh^{-}:=\x-\vx_0,\; \vh^{+}:=\x+\vx_0$ and set
\begin{eqnarray*}
  && T_1:=\left\{j: \mathrm{sign}(\nj{\va_j,\x})=1,\; \mathrm{sign}(\nj{\va_j,\vx_0})=1\right\} \\
  && T_2:=\left\{j: \mathrm{sign}(\nj{\va_j,\x})=-1,\; \mathrm{sign}(\nj{\va_j,\vx_0})=-1\right\}  \\
  && T_3:=\left\{j: \mathrm{sign}(\nj{\va_j,\x})=1,\; \mathrm{sign}(\nj{\va_j,\vx_0})=-1\right\} \\
  && T_4:=\left\{j: \mathrm{sign}(\nj{\va_j,\x})=-1,\; \mathrm{sign}(\nj{\va_j,\vx_0})=1\right\}.
\end{eqnarray*}

  Without loss of generality, we can assume that $\# (T_{1}\cup T_2)=\beta m \ge m/2$.
Using an argument  similar to  one  for (\ref{key eq}), we obtain that
\begin{equation} \label{key inequality for sparse}
  \norm{A_{T_{12}}\vh^{-}}^2 \le 2\nj{\vh^{-},A_{T_1}^\T \eta_{T_1}-A_{T_2}^\T \eta_{T_2}}+ \|\eta_{T_{12}^c}\|^2.
\end{equation}
 To this end, we first need to show $\normone{\vh^{-}}\le 2\sqrt{s}\norm{\vh^{-}}$. Indeed, let $S:=\supp(\vx)$ and note that
\begin{equation*}
  \normone{\x}=\normone{\vx_0+\vh^{-}}=\normone{\vx_0+\vh_{S}^{-}}+\normone{\vh_{S^c}^{-}}\ge \normone{\vx_0}-\normone{\vh_{S}^{-}}+\normone{\vh_{S^c}^{-}}.
\end{equation*}
Here $\vh_{S}^{-}$ denotes the restriction of the vector $\vh^{-}$ onto the set of coordinates $S$. Then the constrain condition $\normone{\x}\le R:=\normone{\vx_0}$ implies that $\normone{\vh_{S^c}^{-}}\le \normone{\vh_{S}^{-}}$. Using H\"older inequality, we obtain that
\begin{equation*}
  \normone{\vh^{-}}\,=\,\normone{\vh_{S}^{-}}+\normone{\vh_{S^c}^{-}}\, \le\, 2\normone{\vh_{S}^{-}}\,\le\, 2\sqrt{s}\norm{\vh^{-}}.
\end{equation*}
We next give a lower bound for the left hand of inequality (\ref{key inequality for sparse}). Set
\begin{equation*}
  K:=\dkh{\vh \in \Rd: \norm{\vh}\le 1, \; \normone{\vh}\le 2\sqrt{s}}.
\end{equation*}
Note that $\vh^{-}/\norm{\vh^{-}}\in K$.
Since $A/\sqrt{m}$ satisfies strong RIP (see Lemma \ref{SRIP}), we obtain that
\begin{equation}\label{lower bound sparse cons}
 \norm{A_{T_{12}}\vh^{-}}^2\,\, \ge\,\, cm \norm{\vh^{-}}^2
\end{equation}
holds  with probability at least $1-\exp(-c_0 m)$, provided $m\gtrsim s\log (ed/s)$.

On the other hand,  Lemma \ref{upper bound sparse cons} implies that
\begin{equation}\label{up:bo}
  \nj{\vh^{-},A_{T_1}^\T \eta_{T_1}-A_{T_2}^\T \eta_{T_2}}\,\, \le\,\, 2C\sqrt{m}\norm{\eta}\norm{\vh^{-}}
\end{equation}
holds with probability at least $1-2\exp(-c_0m)$.
Putting (\ref{up:bo}) and (\ref{lower bound sparse cons})  into (\ref{key inequality for sparse}),  we obtain that
\begin{equation}\label{eq:hmineq}
  cm \norm{\vh^{-}}^2\le 4C\sqrt{m}\norm{\eta}\norm{\vh^-}+ \|\omega_{T_{12}^c}\|^2
\end{equation}
holds with probability at least $1-3\exp(-c_0 m)$. The (\ref{eq:hmineq}) implies that
\[
  \norm{\vh^{-}} \lesssim \frac{\norm{\omega}}{\sqrt{m}}.
\]
Similarly, if $\# (T_3\cup T_4)\geq m/2$, we can obtain that
\[
  \norm{\vh^{+}} \lesssim \frac{\norm{\omega}}{\sqrt{m}}.
\]
\end{proof}

\subsection{Proof of Theorem \ref{th:sparse}}
To this end, we introduce the following lemma.

\begin{lemma} \label{upper bound sparse uncons}
Let $A\in \R^{m\times d}$ be a Gaussian matrix whose entries are independent Gaussian random variables and $\eta \in \R^m$ be a fixed vector.
Then the following holds with probability at least $1-1/d^2$
\begin{equation} \label{constant c}
  \sup_{\vh\in \Rd \atop T\subset \{1,\ldots,m\}}\nj{\vh,A^\T \eta_{T}}\lesssim  ( \normone{\eta}+\norm{\eta}\sqrt{\log d})\normone{\vh},
\end{equation}
where $\eta_T$ denotes the vector generated by $\eta$ with entries in $T$ are themselves and others are zeros.
\end{lemma}
\begin{proof}
By applying H\"older's inequality with $\ell_1$ and $\ell_\infty$ norms, we have
\begin{equation*}
  \nj{\vh,A^\T \eta_{T}} \le  \norms{A^\T \eta_{T}}_\infty\cdot \normone{\vh}.
\end{equation*}
Thus it is sufficient to present an upper bound of  $\sup_{T\subset \{1,\ldots,m\}} \norms{A^\T \omega_{T}}_\infty$.
We use $\ta_j\in \R^{m}, j=1,\ldots,d,$ to denote the {\em column} vectors of $A$.
 Then for any fixed index $j$ and $t>0$, we have
\begin{equation*}
  \PP\xkh{\sup_{T\subset \{1,\ldots,m\}} \abs{\ta_j^\T \omega_{T}}>t}\le  \PP\xkh{\sum_{i=1}^m \abs{\omega_{i}}\abs{\ta_{j,i}} >t}.
\end{equation*}
A simple calculation shows that $\E \abs{\omega_i}\abs{\ta_{j,i}}=\sqrt{2/\pi}\abs{\omega_i}$. By  Hoeffding's inequality, we  obtain that
\begin{equation}\label{eq:pror}
  \PP\xkh{\sum_{i=1}^m \abs{\omega_{i}}\abs{\ta_{j,i}} > C\Big(\normone{\eta}+\norm{\eta}\sqrt{\log d}\Big)} \le \frac{1}{d^3}
\end{equation}
holds for some constant $C>0$.
Taking a union bound over all indexes $j\in \{1,\ldots,d\}$, (\ref{eq:pror}) implies
\begin{equation*}
  \sup_{T\subset \{1,\ldots,m\}} \norms{A^\T \omega_{T}}_\infty \lesssim \normone{\eta}+\norm{\eta}\sqrt{\log d}
\end{equation*}
with probability at least $1-1/d^2$. Thus, we arrive at the conclusion.
\end{proof}

\begin{proof}[Proof of Theorem \ref{th:sparse}]
Set $\vh^{-}:=\x-\vx_0$ and $\vh^{+}:=\x+\vx_0$. Without loss of generality, we assume that $\normone{\vh^{-}}\le \normone{\vh^{+}}$. Since $\x$ is the solution  of (\ref{square model with sparse}), we have
\begin{equation} \label{minsolution sparse}
  \left\||A\x|-\vy\right\|^2+\lambda\norms{\x}_1 \le \left\||A\vx_0|-\vy\right\|^2+\lambda\norms{\vx_0}_1 =\left\|\eta\right\|_2^2+\lambda\norms{\vx_0}_1.
\end{equation}
For any index set $T\subset \{1,\ldots,m\}$, we set $A_T:=[\va_j:\;j\in T]^\T$ which is a submatrix of $A$. Set
\begin{eqnarray*}
  && T_1:=\left\{j: \mathrm{sign}(\nj{\va_j,\x})=1,\; \mathrm{sign}(\nj{\va_j,\vx_0})=1\right\}  \\
  && T_2:=\left\{j: \mathrm{sign}(\nj{\va_j,\x})=-1,\; \mathrm{sign}(\nj{\va_j,\vx_0})=-1\right\}  \\
  && T_3:=\left\{j: \mathrm{sign}(\nj{\va_j,\x})=1,\; \mathrm{sign}(\nj{\va_j,\vx_0})=-1\right\} \\
  && T_4:=\left\{j: \mathrm{sign}(\nj{\va_j,\x})=-1,\; \mathrm{sign}(\nj{\va_j,\vx_0})=1\right\}.
\end{eqnarray*}
Then a simple calculation leads to
\begin{equation} \label{Ax-y}
 \left\||A\x|-\vy\right\|^2=\norm{A_{T_1}\vh^{-}-\omega_{T_1}}^2+\norm{A_{T_2}\vh^{-}+\omega_{T_2}}^2+\norm{A_{T_3}\vh^{+}-\omega_{T_3}}^2+\norm{A_{T_4}\vh^{+}+\omega_{T_4}}^2.
\end{equation}
Substituting (\ref{Ax-y})   into (\ref{minsolution sparse}), we obtain that
\begin{eqnarray} \label{key eq sparse}
\norm{A_{T_{12}}\vh^{-}}^2+\norm{A_{T_{34}}\vh^{+}}^2  &\le & 2\nj{\vh^{-},A_{T_1}^\T \omega_{T_1}-A_{T_2}^\T \omega_{T_2}}+2\nj{\vh^{+},A_{T_3}^\T \omega_{T_3}-A_{T_4}^\T \omega_{T_4}} \nonumber \\
   & & + \lambda(\norms{\vx_0}_1-\norms{\vh^{+}-\vx_0}_1),
\end{eqnarray}
where $T_{12}:=T_1 \cup T_2$ and $T_{34}:=T_3 \cup T_4$. We claim that $\normone{\vh^{-}} \le 4\sqrt{s}\norm{\vh^{-}}$ and $\normone{\vh^{+}} \le 4\sqrt{s}\norm{\vh^{+}}$ hold with high probability. Indeed, let $S:=\supp(\vx_0)\subset \{1,\ldots,d\}$. Then
\begin{equation}\label{normoneh}
  \norms{\vh^{+}-\vx_0}_1 =\normone{\vh_S^{+}-\vx_0}+\normone{\vh_{S^c}^{+}} \ge \normone{\vx_0}-\normone{\vh_S^{+}}+\normone{\vh_{S^c}^{+}},
\end{equation}
where the inequality follows from triangle inequality. According to Lemma \ref{upper bound sparse uncons}, we obtain that
\begin{equation} \label{eq:up_lemma}
  \nj{\vh^{-},A_{T_1}^\T \omega_{T_1}-A_{T_2}^\T \omega_{T_2}}\le \frac{\lambda}{8}\normone{\vh^{-}} \quad \mathrm{and}\quad \nj{\vh^{+},A_{T_3}^\T \omega_{T_3}-A_{T_4}^\T \omega_{T_4}}\le \frac{\lambda}{8} \normone{\vh^{+}}
\end{equation}
 holds with probability at least $1-1/d^2$, where $\lambda \gtrsim  \normone{\eta}+\norm{\eta}\sqrt{\log d}$. Putting (\ref{normoneh}) and (\ref{eq:up_lemma}) into (\ref{key eq sparse}) and using the fact $\normone{\vh^{-}}\le \normone{\vh^{+}}$,  we can obtain that
\begin{equation}\label{eq:AT12H}
  \norm{A_{T_{12}}\vh^{-}}^2+\norm{A_{T_{34}}\vh^{+}}^2 \le  \frac{\lambda}{2} \normone{\vh^{+}}+ \lambda(\normone{\vh_S^{+}}-\normone{\vh_{S^c}^{+}})
\end{equation}
holds with probability at least $1-1/d^2$.  The (\ref{eq:AT12H}) implies that
\begin{equation*}
 \frac{\lambda}{2} \normone{\vh^{+}}+ \lambda(\normone{\vh_S^{+}}-\normone{\vh_{S^c}^{+}})\ge 0,
\end{equation*}
which gives $\normone{\vh_{S^c}^{+}} \le 3\normone{\vh_S^{+}}$ and hence
$\normone{\vh^{+}} \le 4\normone{\vh_S^{+}}$. By the H\"older's inequality, we obtain that
 \[
 \normone{\vh^{+}} \le 4\sqrt{s}\norm{\vh^{+}}.
 \]
    On the other hand, note that
\begin{equation*}
  \normone{\vh_S^{+}}=\normone{\x_S+\vx_0}, \quad \normone{\vh_S^{-}}=\normone{\x_S-\vx_0}\quad \mathrm{and} \quad \normone{\vh_{S^c}^{+}}=\normone{\vh_{S^c}^{-}}.
\end{equation*}
Combining with $\normone{\vh^{-}}\le \normone{\vh^{+}}$, we can  obtain that $\normone{\vh^{-}} \le 4\sqrt{s}\norm{\vh^{-}}$.

 We next present an upper bound of $\norm{\vh^{-}}$. Without loss of generality, we assume that $\# T_{12}=\beta m \ge m/2$. The  (\ref{Ax-y}) implies that
\begin{equation}\label{eq:jin}
  \left\||A\x|-\vy\right\|^2\ge \norm{A_{T_1}\vh^{-}-\omega_{T_1}}^2+\norm{A_{T_2}\vh^{-}+\omega_{T_2}}^2.
\end{equation}
Substituting (\ref{eq:jin})  into (\ref{minsolution sparse}) we obtain that
\begin{equation}\label{key}
\begin{aligned}
  \norm{A_{T_{12}}\vh^{-}}^2 &\le  2\nj{\vh^{-},A_{T_1}^\T \omega_{T_1}-A_{T_2}^\T \omega_{T_2}} + \lambda(\norms{\vx_0}_1-\norms{\vh^{-}+\vx_0}_1)+ \|\omega_{T_{12}^c}\|^2\\
&\le 2\nj{\vh^{-},A_{T_1}^\T \omega_{T_1}-A_{T_2}^\T \omega_{T_2}} + \lambda(\normone{\vh_S^{-}}-\normone{\vh_{S^c}^{-}})+ \|\omega_{T_{12}^c}\|^2.
\end{aligned}
\end{equation}
Here, we use
\begin{equation*}
  \norms{\vh^{-}+\vx_0}_1 =\normone{\vh_S^{-}+\vx_0}+\normone{\vh_{S^c}^{-}} \ge \normone{\vx_0}-\normone{\vh_S^{-}}+\normone{\vh_{S^c}^{-}}.
\end{equation*}
We  consider the left side of (\ref{key}).
Recall that $\normone{\vh^{-}} \le 4\sqrt{s}\norm{\vh^{-}}$. Then
\begin{equation}\label{lower bound sparse uncons}
 \norm{A_{T_{12}}\vh^{-}}^2 \ge cm \norm{\vh^{-}}^2
\end{equation}
 with probability at least $1-\exp(-c_0 m)$, provided $m \gtrsim s\log (ed/s)$
(see Remark \ref{re:SRIP}).
For the right hand of (\ref{key}), we use  (\ref{eq:up_lemma}) to obtain that
\begin{equation}\label{sparse:upper bound}
\begin{aligned}
 \norm{A_{T_{12}}\vh^{-}}^2 &\le \frac{\lambda}{4}\normone{\vh^{-}}+\lambda(\normone{\vh_S^{-}}-\normone{\vh_{S^c}^{-}})+ \|\omega_{T_{12}^c}\|^2 \\
   &\le  \frac{5\lambda}{4}\normone{\vh_S^{-}}+ \|\omega_{T_{12}^c}\|^2 \\
   &\le   \frac{5\lambda\sqrt{s}}{4}\norm{\vh^{-}}+ \|\omega_{T_{12}^c}\|^2
\end{aligned}
\end{equation}
holds with probability at least $1-1/d^2$. Combining (\ref{lower bound sparse uncons}) and (\ref{sparse:upper bound}),  we have
\begin{equation*}
  cm \norm{\vh^{-}}^2\le \frac{5\lambda\sqrt{s}}{4}\norm{\vh^{-}}+ \|\omega_{T_{12}^c}\|^2
\end{equation*}
 with probability at least $1-\exp(-c_0 m)-1/d^2$. By solving the above inequality, we arrive at the conclusion
\begin{eqnarray*}
  \norm{\vh^{-}} &\lesssim & \frac{\lambda\sqrt{s}}{m}+\frac{\norm{\omega}}{\sqrt{m}}.
\end{eqnarray*}
\end{proof}

\section{Discussion}
We have analyzed the estimation performance of the nonlinear least squares for phase retrieval. We show that the reconstruction error of the nonlinear least square model is $O(\|\eta\|_2/\sqrt{m})$ and we also prove that this recovery bound is optimal up to a constant.  For sparse phase retrieval, we also obtain similar results  for the nonlinear Lasso.
It is of interest to extend the results in this paper to complex signals. Moreover, assume that
 $y_i=f(\abs{\va_i,\vx_0})+\eta_i, i=1,\ldots,m$, where $f:\mathbb{R}\to \mathbb{R}$ is a continuous function. It is interesting to consider the recovery error of the model $\min_\vx \|\abs{A\vx}-\vy\|$ under this setting, which is the subject of our future work.


\begin{thebibliography}{10}

\bibitem{cai2016optimal}
T~Tony Cai, Xiaodong Li, Zongming Ma, et~al.
\newblock Optimal rates of convergence for noisy sparse phase retrieval via
  thresholded wirtinger flow.
\newblock {\em The Annals of Statistics}, 44(5):2221--2251, 2016.


\bibitem{WF}
Emmanuel~J Candes, Xiaodong Li, and Mahdi Soltanolkotabi.
\newblock Phase retrieval via wirtinger flow: Theory and algorithms.
\newblock {\em IEEE Transactions on Information Theory}, 61(4):1985--2007,
  2015.

\bibitem{phaselift}
Emmanuel~J Candes, Thomas Strohmer, and Vladislav Voroninski.
\newblock Phaselift: Exact and stable signal recovery from magnitude
  measurements via convex programming.
\newblock {\em Communications on Pure and Applied Mathematics},
  66(8):1241--1274, 2013.

\bibitem{TWF}
Yuxin Chen and Emmanuel Candes.
\newblock Solving random quadratic systems of equations is nearly as easy as
  solving linear systems.
\newblock In {\em Advances in Neural Information Processing Systems}, pages
  739--747, 2015.


\bibitem{fienup1987phase}
C.~Fienup and J.~Dainty.
\newblock Phase retrieval and image reconstruction for astronomy.
\newblock {\em Image Recovery: Theory and Application}, pages 231--275, 1987.


\bibitem{ER3}
J. ~R. Fienup. Phase retrieval algorithms: a comparison.
\newblock {\em Applied optics}, 21(15):2758--2769, 1982.

\bibitem{Gaoxu}
Bing Gao and Zhiqiang Xu.
\newblock Phaseless recovery using the Gauss--Newton method.
\newblock {\em IEEE Transactions on Signal Processing}, 65(22):5885--5896,
  2017.


\bibitem{gerchberg1972practical}
Ralph~W Gerchberg.
\newblock A practical algorithm for the determination of phase from image and
  diffraction plane pictures.
\newblock {\em Optik}, 35:237, 1972.

\bibitem{phasemax}
Tom Goldstein and Christoph Studer.
\newblock Phasemax: Convex phase retrieval via basis pursuit.
\newblock {\em IEEE Transactions on Information Theory}, 64(4):2675--2689,
  2018.

\bibitem{harrison1993phase}
Robert~W Harrison.
\newblock Phase problem in crystallography.
\newblock {\em JOSA A}, 10(5):1046--1055, 1993.




\bibitem{HX}
Meng Huang and Zhiqiang Xu.
\newblock Phase retrieval from the norms of affine transformations.
\newblock {\em arXiv preprint arXiv:1805.07899 }, 2018.

\bibitem{Iwen2015Robust}
Mark Iwen, Aditya Viswanathan, and Yang Wang.
\newblock Robust sparse phase retrieval made easy.
\newblock {\em Applied and Computational Harmonic Analysis}, 42(1):135--142,
  2015.


\bibitem{miao2008extending}
Jianwei Miao, Tetsuya Ishikawa, Qun Shen, and Thomas Earnest.
\newblock Extending x-ray crystallography to allow the imaging of
  noncrystalline materials, cells, and single protein complexes.
\newblock {\em Annu. Rev. Phys. Chem.}, 59:387--410, 2008.

\bibitem{millane1990phase}
Rick~P Millane.
\newblock Phase retrieval in crystallography and optics.
\newblock {\em JOSA A}, 7(3):394--411, 1990.

\bibitem{AltMin}
Praneeth Netrapalli, Prateek Jain, and Sujay Sanghavi.
\newblock Phase retrieval using alternating minimization.
\newblock {\em IEEE Transactions on Signal Processing}, 63(18):4814--4826,
  2015.


\bibitem{Plan2013}
Yaniv Plan and Roman Vershynin.
\newblock One-bit compressed sensing by linear programming.
\newblock {\em Communications on Pure and Applied Mathematics},
  66(8):1275--1297, 2011.


\bibitem{plan2016generalized}
Yaniv Plan and Roman Vershynin.
\newblock The generalized lasso with non-linear observations.
\newblock {\em IEEE Transactions on information theory}, 62(3):1528--1537,
  2016.

\bibitem{shechtman2015phase}
Yoav Shechtman, Yonina~C Eldar, Oren Cohen, Henry~Nicholas Chapman, Jianwei
  Miao, and Mordechai Segev.
\newblock Phase retrieval with application to optical imaging: a contemporary
  overview.
\newblock {\em IEEE signal processing magazine}, 32(3):87--109, 2015.

\bibitem{turstregion}
Ju~Sun, Qing Qu, and John Wright.
\newblock A geometric analysis of phase retrieval.
\newblock In {\em IEEE International Symposium on Information Theory}, pages
  1--68, 2016.

\bibitem{tan2017phase}
Yan~Shuo Tan and Roman Vershynin.
\newblock Phase retrieval via randomized kaczmarz: theoretical guarantees.
\newblock {\em Information and Inference: A Journal of the IMA}, 2017.

\bibitem{Vershynin2018}
Roman Vershynin.
\newblock {\em High-dimensional probability: An introduction with applications
  in data science}, volume~47.
\newblock Cambridge University Press, 2018.

\bibitem{Voroninski2016A}
Vladislav Voroninski and Zhiqiang Xu.
\newblock A strong restricted isometry property, with an application to
  phaseless compressed sensing .
\newblock {\em Applied and Computational Harmonic Analysis}, 40(2):386--395,
  2016.

\bibitem{ER4}
Ir\`{e}ne Waldspurger,
\newblock Phase retrieval with random Gaussian sensing vectors by alternating projections,
\newblock {\em IEEE Transactions on Information Theory}, 2018.


\bibitem{walther1963question}
Adriaan Walther.
\newblock The question of phase retrieval in optics.
\newblock {\em Journal of Modern Optics}, 10(1):41--49, 1963.

\bibitem{TAF}
Gang Wang, Georgios~B. Giannakis, and Yonina~C. Eldar.
\newblock Solving systems of random quadratic equations via truncated amplitude
  flow.
\newblock {\em IEEE Transactions on Information Theory}, 64(2):773--794, 2018.

\bibitem{SparseTAF}
Gang Wang, Liang Zhang, Georgios~B. Giannakis, Mehmet Akcakaya, and Jie Chen.
\newblock Sparse phase retrieval via truncated amplitude flow.
\newblock {\em IEEE Transactions on Signal Processing}, PP(99):1--1, 2016.

\bibitem{wang2016generalized}
Yang Wang and Zhiqiang Xu.
\newblock Generalized phase retrieval: measurement number, matrix recovery and
  beyond.
\newblock {\em Applied and Computational Harmonic Analysis}, 2017.


\bibitem{zhang2016reshaped}
Huishuai Zhang and Yingbin Liang.
\newblock Reshaped wirtinger flow for solving quadratic system of equations.
\newblock In {\em Advances in Neural Information Processing Systems}, pages
  2622--2630, 2016.


\end{thebibliography}
\end{document}